\documentclass[a4paper, 11pt]{article}

\usepackage{amsmath,amssymb,amsthm}

\usepackage[a4paper, margin=1.2in]{geometry}

\newtheorem{theorem}{Theorem}
\newtheorem{proposition}[theorem]{Proposition}
\newtheorem{lemma}[theorem]{Lemma}
\newtheorem{corollary}[theorem]{Corollary}
\newtheorem{definition}[theorem]{Definition}

\usepackage{makeidx}
\usepackage{textcomp}
\usepackage{graphicx,graphics}
\usepackage{psfrag}
\usepackage{epsfig}

\newcommand{\sbs}{\texttt{\textup{searchball}}}
\newcommand{\sbf}{\texttt{\textup{searchball-fast}}}

\newcommand{\vbl}{\ensuremath{{\rm vbl}}}

\newcommand{\poly}{\ensuremath{{\textup{\rm poly}}}}

\newcommand{\vol}[3]{\ensuremath{{\rm vol}^{(#1)}(#2,#3)}}
\newcommand{\ball}[3]{\ensuremath{B^{(#1)}_{#2}(#3)}}

\newcommand{\schoening}{\texttt{\textup{Sch\"oning}}}

\newcommand{\ignore}[1]{}

\usepackage{cancel}
\usepackage{algorithm}
\usepackage{algorithmic}

\newcommand{\PBkS}[1]{\ensuremath{\textsc{Promise-Ball-}#1\textsc{-SAT}}}
\newcommand{\BkS}[1]{\ensuremath{\textsc{Ball-}#1\textsc{-SAT}}}

\title{A Full Derandomization of
  Sch\"oning's $k$-SAT Algorithm}

\author{Robin A. Moser and Dominik Scheder \vspace{3mm}\\
Institute for Theoretical Computer Science\\
Department of Computer Science\\
ETH Z\"urich, 8092 Z\"urich, Switzerland\\
\texttt{\{robin.moser, dominik.scheder\}@inf.ethz.ch}
}

\begin{document}

\maketitle

\begin{abstract}
  Sch\"oning~\cite{Schoening99} presents a simple randomized algorithm
  for $k$-SAT with running time $O(a_k^n \poly(n))$ for 
  $a_k = 2(k-1)/k$. We give a deterministic version of this algorithm
  running in time $O( (a_k+\epsilon)^n \poly(n))$, where
  $\epsilon > 0$ can be made arbitrarily small.
\end{abstract}

\section{Introduction}

In 1999, Uwe Sch\"oning~\cite{Schoening99} gave an extremely simple
randomized algorithm for $k$-SAT. Ten years on, the fastest
algorithms for $k$-SAT are only slightly faster than his, and far more
complicated. His algorithm works as follows: Let $F$ be a $(\le k)$-CNF
formula over $n$ variables. Start with a random truth assignment. If
this does not satisfy $F$, pick an arbitrary unsatisfied clause $C$.
From $C$, pick a literal uniformly at random, and change the truth
value of its underlying variable, thus satisfying $C$.  Repeat this
reassignment step $O(n)$ times.  If $F$ is satisfiable, this finds a
satisfying assignment with probability at least
$$
\left(\frac{k}{2(k-1)}\right)^n \ .
$$
By repetition, this gives a randomized $O^*(1.334^n)$ algorithm for
$3$-SAT, an $O^*(1.5^n)$ for $4$-SAT, and so on (we use $O^*$ to
suppress polynomial factors in $n$). Shortly after Sch\"oning
published his algorithm, Dantsin, Goerdt, Hirsch, Kannan, Kleinberg,
Papadimitriou, Raghavan and Sch\"oning~\cite{dantsin} (henceforth {\em
  Dantsin et al.} for the sake of brevity) came up with a {\em
  deterministic} algorithm that can be seen as an attempt to
derandomize Sch\"oning's algorithm. We say {\em attempt} because its
running time is $O^*( (2k / (k+1))^n)$, which is exponentially slower
than Sch\"oning's. For example, this gives an $O^*(1.5^n)$ algorithm
for $3$-SAT and $O^*(1.6^n)$ for $4$-SAT. Subsequent papers have
improved upon this running time, mainly focusing on $3$-SAT: Dantsin
et al. already improve the running time for $3$-SAT to $O(1.481^n)$,
Brueggemann and Kern~\cite{BK04} to $O(1.473^n)$,
Scheder~\cite{Scheder08} to $O(1.465^n)$, and Kutzkov and
Scheder~\cite{kutzkov-scheder} to $O^*(1.439^n)$. All improvements
suffer from two drawbacks: First, they fall short of achieving the
running time of Sch\"oning's randomized algorithm, and second, they
are all fairly complicated. In this paper, we give a rather simple
deterministic algorithm with a running time that comes arbitrarily
close to Sch\"oning's, thus completely derandomizing his algorithm.
We also show how to derandomize Sch\"oning's algorithm for constraint
satisfaction problems, which are a generalization of SAT, allowing
more than two truth values.

\subsection{Notation}

We use the notational framework introduced in \cite{welzl05}. We
assume an infinite supply of propositional \emph{variables}. A
\emph{literal} $u$ is a variable $x$ or a complemented variable $\bar
x$. A finite set $C$ of literals over pairwise distinct variables is
called a \emph{clause} and a finite set of clauses is a \emph{formula}
in CNF (Conjunctive Normal Form). We say that a variable $x$
\emph{occurs} in a clause $C$ if either $x$ or $\bar x$ are contained
in it and that $x$ occurs in the formula $F$ if there is any clause
where it occurs. We write $\mbox{vbl}(C)$ or $\mbox{vbl}(F)$ to denote
the set of variables that occur in $C$ or in $F$, respectively.  We
say that $F$ is a $(\le k)$\emph{-CNF formula} if every clause has
size at most $k$.  Let such an $F$ be given and write
$V:=\mbox{vbl}(F)$.

A \emph{(truth) assignment} is a function $\alpha : V \rightarrow
\{0,1\}$ which assigns a Boolean value to each variable. A literal
$u=x$ (or $u=\bar x$) is \emph{satisfied by} $\alpha$ if $\alpha(x)=1$
(or $\alpha(x)=0$). A clause is \emph{satisfied by} $\alpha$ if it
contains a satisfied literal and a formula is \emph{satisfied by}
$\alpha$ if all of its clauses are. A formula is \emph{satisfiable} if
there exists a satisfying truth assignment to its variables.

If $\alpha$ and $\beta$ are two truth assignments over a set $V$ of
variables, then their \emph{(Hamming) distance} $d_H(\alpha,\beta)$ is
defined to be the number of variables $x \in V$ where $\alpha(x) \ne
\beta(x)$, i.e.  $d_H(\alpha,\beta) := |\{ x \in V \; | \; \alpha(x)
\ne \beta(x) \}|$.  For a given assignment $\alpha$, we denote the set
of all assignments $\beta$ with Hamming distance at most $r$ from
$\alpha$ by $B_r(\alpha) := \{ \beta : V \rightarrow \{0,1\} \ | \
d_H(\alpha,\beta) \le r \}$ and call this the \emph{Hamming ball of
  radius $r$ centered at $\alpha$}.

Formulas can be manipulated by permanently assigning values to
variables. If $F$ is a given CNF formula and $x \in \mbox{vbl}(F)$
then assigning $x \mapsto 1$ satisfies all clauses containing $x$
(irrespective of what values the other variables in those closes are
possibly assigned later) whilst it truncates all clauses containing
$\bar x$ to their remaining literals. We will write $F^{[x:=1]}$ to
denote the formula arising from doing just this, or equally
$F^{[u:=1]}$ where $u$ is a literal and we mean to assign the
underlying variable the value necessary to satisfy $u$. If $\beta$ is
a partial assignment, i.e., defined on a subset of $\vbl(F)$, then
$F^{[\beta]}$ denotes the formula we obtain from $F$ by permanently
setting the variables from those subset to their respective values
under $\beta$.

\subsection{Previous Work}

Both Sch\"oning's algorithm and its deterministic versions can be 
seen as not attacking SAT directly, but rather a parametrized
local search problem:

\begin{quotation}
  \textbf{$\PBkS{k}$:} Given a $(\le k)$-CNF formula $F$ over $n$
  variables, an assignment $\alpha$ to these variables, a natural
  number $r$, and the promise that the Hamming ball $B_r(\alpha)$
  contains a satisfying assignment.  Find any satisfying assignment
  to $F$.
\end{quotation}

Let us clarify what we mean by saying ``Algorithm $A$ solves
$\PBkS{k}$'': If $F$, $\alpha$, and $r$ are as described above, i.e.,
if $B_r(\alpha)$ contains a satisfying assignment, then $A$ must
return {\em some} satisfying assignment. We do not require this
assignment to lie in $B_r(\alpha)$, however.  On the other hand, if
$F$ is unsatisfiable, or $B_r(\alpha)$ contains no satisfying
assignment, the behavior is unspecified. Of course, since we can
quickly check any purported assignment that the algorithm outputs, we
can assume the algorithm always either returns a satisfying assignment
or \texttt{failure}.

Sch{\"o}ning's original random\-ized al\-go\-rithm for $\PBkS{k}$ as
described in the introductory section, henceforth called $\schoening$,
repeatedly selects any clause unsatisfied under $\alpha$, then
randomly picks a literal from that clause and flips the underlying
variable's value.  The algorithm gives up if a satisfying assignment
has not been encountered by the time $O(n)$ steps have been performed
(it is well-known and easy to check that $n/(k-2)$ correction steps
are sufficient to achieve optimal efficiency).

\begin{lemma}[Sch\"oning~\cite{Schoening99}]
  Let $F$ be a $(\le k)$-CNF formula, $\alpha$ a truth assignment to its
  variables, and $r \in \mathbb{N}$. If there is a satisfying
  assignment in $B_r(\alpha)$, then with probability at least
  $(k-1)^{-r}$, $\schoening$ returns a satisfying assignment.  By
  repetition, this gives a Monte-Carlo algorithm for $\PBkS{k}$ with
  running time $O^*((k-1)^r)$.
\label{lemma-2-r-random}
\end{lemma}
Sch\"oning turns this lemma into an algorithm for $k$-SAT by
choosing the assignment $\alpha$ uniformly at random from
all $2^n$ truth assignments:
\begin{theorem}[Sch\"oning~\cite{Schoening99}]
  There is a randomized algorithm that runs in polynomial
  time and finds a satisfying assignment of $F$ with probability
  $$
  \left(\frac{k}{2(k-1)}\right)^n \ , 
  $$
  provided $F$ is satisfiable.
  \label{theorem-schoening}
\end{theorem}
\begin{proof}
  Let $\alpha^*$ be a satisfying assignment of $F$ and let $\alpha$ be
  an assignment chosen uniformly at random from $\{0,1\}^n$.  For each
  $0 \leq r \leq n$, the probability that the Hamming distance
  $d_H(\alpha,\alpha^*)$ is $r$ is ${n \choose r} / 2^n$.  In this
  case, Sch\"oning's random walk returns a satisfying assignment with
  probability at least $(k-1)^{-r}$. The overall success probability
  thus is at least
  $$
  \sum_{r=0}^n {n \choose r} 2^{-n} (k-1)^{-r} 
  = \left(\frac{k}{2(k-1)}\right)^n \ ,
  $$
  and the running time is clearly polynomial.
\end{proof}
By repeating the above algorithm, one obtains a Monte-Carlo algorithm
for $k$-SAT of running time $O^*( (2(k-1)/k)^n)$.

\subsection*{Deterministic Algorithms}

What about deterministic algorithms? Dantsin et al.~\cite{dantsin}
give a simple recursive algorithm for $\PBkS{k}$ running in time
$O^*(k^r)$: If $\alpha$ satisfies $F$, we are done. Otherwise, if
$r=0$, we can return \texttt{failure}. If $r \geq 1$ and $\alpha$ does
not satisfy $F$, we let $C$ be an unsatisfied clause. There are at
most $k$ literals in $C$, thus there are at most $k$ possibilities to
locally change $\alpha$ so as to satisfy $C$. We recursively explore each
possibility, decreasing $r$ by $1$ (see Algorithm~\ref{sbs} for the details). 
The next achievement of Dantsin et al. is to show how a deterministic algorithm
for $\PBkS{k}$ can be turned into a deterministic algorithm for
$k$-SAT:
\begin{lemma}[Dantsin et al.~\cite{dantsin}]
  If algorithm A solves $\PBkS{k}$ in time $O^*(a^r)$, then there is
  an algorithm $B$ solving $k$-SAT in time
  $O^*\left(\left(\frac{2a}{a+1}\right)^n\right)$. Furthermore, $B$ is
  deterministic if $A$ is.
\label{ball-to-cube}
\end{lemma}
Their algorithm to prove the lemma constructs a so-called covering
code $\mathcal C \subseteq \{0,1\}^n$ with the property that every
assignment $\alpha \in \{0,1\}^n$ has a codeword $\gamma \in
\mathcal{C}$ at a suitably small Hamming distance from $\alpha$.
Sch{\"o}ning's randomized selection of an initial assignment is turned
deterministic by iterating through all codewords $\gamma \in
\mathcal{C}$ and solving $\PBkS{k}$ around each of them. Provided that
the formula is satisfiable, one choice of $\gamma \in \mathcal C$ will
be sufficiently close to a satisfying assignment for the subsequent
local search to succeed.

The recursive algorithm for $\PBkS{k}$ of Dantsin et al. has running
time $O^*(k^r)$.  Therefore Lemma~\ref{ball-to-cube} gives a running
time of $O^*( (2k/(k+1))^n)$. For $k=3$, clever branching rules have
been designed to improve upon the $O^*(3^r)$ bound, leading to the
respective improvements on deterministic running times mentioned in
the first paragraph of this paper.

\subsection{Our Contribution}

Our contribution is to give a deterministic algorithm solving
$\PBkS{k}$ in a running time that gets arbitrarily close to
that of the Monte-Carlo algorithm in Lemma~\ref{lemma-2-r-random}.
\begin{theorem}
  For every $\epsilon > 0$, there exists a deterministic algorithm
  which solves the problem $\PBkS{k}$ in time $O( (k-1+\epsilon)^r)$.
\label{theorem-det-walk}
\end{theorem}
Combining this theorem with Lemma~\ref{ball-to-cube} proves
our main theorem:
\begin{theorem}
  For every $\epsilon > 0$, there is a deterministic algorithm solving
  $k$-SAT in time $O^*\left( \left(\frac{2(k-1)}{k}+\epsilon\right)^n
  \right)$.
\label{main-theorem}
\end{theorem}

Before jumping into technical details, let us sketch the main idea of
our improvement for $k=3$.  Let $F$ be a $3$-CNF
formula and $\alpha$ some assignment. Suppose $F$ contains $t$
pairwise disjoint clauses $C_1,\dots,C_t$, all of which are
unsatisfied by $\alpha$. We let Sch\"oning's random walk algorithm
process these clauses one after the other: In each clause $C_i$, it
picks one literal randomly and satisfies it. Thus, of all $3^t$
possibilities to choose one literal in each $C_i$, it chooses one
uniformly at random. Let $\alpha^*$ be an assignment satisfying $F$.
With probability at least $3^{-t}$, Sch\"oning's random walk chooses
in each $C_i$ a literal that $\alpha^*$ satisfies. In this case, the
distance from $\alpha$ to $\alpha^*$ decreases by $t$. However, with
much bigger probability, roughly $2^{-t/3}$, the random walk chooses
the ``correct'' literal in $2t/3$ clauses $C_i$ and a ``wrong''
literal in the remaining $t/3$.  In this case, the distance from
$\alpha$ to $\alpha^*$ decreases by $t/3$.  This is the power of
Sch\"oning's algorithm: It hopes to make a modest progress of $t/3$,
which is much more likely than making a progress of $t$. Our key
observation is that this choice of Sch\"oning can be derandomized:
There is a set of (roughly) $2^{t/3}$ choices which literal to satisfy
in each $C_i$, such that at least one of them makes a progress of at
least $t/3$.

\section{The Algorithm}

To begin with, we will formally state the recursive algorithm
by Dantsin et al.~\cite{dantsin} solving $\PBkS{k}$ in time
$O^*(k^r)$.
\begin{algorithm}
\caption{\sbs(CNF formula $F$, assignment $\alpha$, radius $r$)}
\label{sbs}
\begin{algorithmic}[1]
  \IF{$\alpha$ satisfies $F$} \RETURN \texttt{true}
  \ELSIF{$r = 0$} \RETURN \texttt{false}
  \ELSE
  \STATE $C \leftarrow $ any clause of $F$ unsatisfied by $\alpha$
  \label{line-select-clause}
  \FOR{$u \in C$}
  \IF{$\sbs(F^{[u := 1]}, 
    \alpha, r-1) = \texttt{true}$} \RETURN \texttt{true}
  \ENDIF
  \ENDFOR
  \RETURN \texttt{false}
  \ENDIF
\end{algorithmic}
\end{algorithm}

\begin{proposition}
  Algorithm~\ref{sbs} solves $\BkS{k}$ in time $O^*(k^r)$
\end{proposition}
\begin{proof}
  The running time is easy to analyze: If $F$ is a $(\le k)$-CNF formula,
  then each call to $\sbs$ causes at most $k$ recursive calls. To see
  correctness of the algorithm, we proceed by induction on $r$ and 
  suppose that $\alpha^*$ satisfies $F$ and $d_H(\alpha,\alpha^*) \leq r$. 
  Let $C$ be the clause selected in
  line~\ref{line-select-clause}.  Since $\alpha^*$ satisfies $C$ but
  $\alpha$ does not, there is at least one literal $u \in C$ such that
  $\alpha^*(u) = 1$ and $\alpha(u)=0$. Let $\alpha' := \alpha^*[u :=
  0]$.  We observe that $d(\alpha,\alpha') \leq r-1$ and $\alpha'$
  satisfies $F^{[u:=1]}$ (although not necessarily $F$). Therefore the
  induction hypothesis ensures that the recursive call to 
  $\sbs(F^{[u:=1]},\alpha,r-1)$ return \texttt{true}.
\end{proof}  
\begin{proposition}
  Suppose $F$ is a $(\le k)$-CNF formula, $\alpha$ a truth assignment to its
  variables, and $r \in \mathbb{N}$. If every clause in $F$ that is
  unsatisfied by $\alpha$ has size at most $k-1$, then
  $\sbs(F,\alpha,r)$ runs in time $O^*( (k-1)^r)$.
\label{prop-short-neg}
\end{proposition}
\begin{proof}
  The key observation is that if all clauses in $F$ that are not
  satisfied by $\alpha$ have at most $k-1$ literals, then the same is
  true for any formula of the form $F^{[u := 1]}$. Therefore,
  any call to $\sbs$ entails at most $k-1$ recursive calls.
\end{proof}

\subsection{$k$-ary Covering Codes}

Before explaining our algorithm, we make a combinatorial detour to
$k$-ary covering codes, which will play a crucial role
in our algorithm.\\

The set $\{1,\dots,k\}^t$ looks similar to the Boolean cube
$\{0,1\}^t$ in many ways. For example, it is endowed with a Hamming
distance $d_H$: For two elements $w,w'\in \{1,\dots,k\}^t$, we define
$d_H(w,w')$ to be the number of coordinates in which $w$ and $w'$ do
not agree. There are also balls: We define
$$
B^{(k)}_r(w) := \{w' \in \{1,\dots,k\}^t \ | \ d_H(w,w') \leq r\} \ .
$$
What is the volume of such a ball? Well, there are ${t \choose r}$
possibilities to choose the set of coordinates in which $w$ and
$w'$ are supposed to differ, and for each such coordinate,
there are $k-1$ ways in which they can differ. Therefore,
$$
\vol{k}{t}{r} := |\ball{k}{r}{w}| = {t \choose r} (k-1)^r \ .
$$
We are interested in the question how many balls $\ball{k}{r}{w}$ we
need to cover all of $\{1,\dots,k\}^t$. Note that by symmetry, $w \in
\ball{k}{r}{v}$ iff $v \in \ball{k}{r}{w}$ for any $v,w \in
\{1,\dots,k\}^t$.
\begin{definition}
  Let $t \in \mathbb{N}$. A set $\mathcal{C} \subseteq
  \{1,\dots,k\}^t$ is called a {\em code of covering radius} $r$ if  
  $$
  \bigcup_{w \in \mathcal{C}} \ball{k}{r}{w} = \{1,\dots,k\}^t \ . 
  $$
  In other words, for each $w' \in \{1,\dots,k\}^n$, there is some $w
  \in \mathcal{C}$ such that $d_H(w,w') \leq r$.
\end{definition}
The following lemma is an adaptation of a lemma by Dantsin et
al.~\cite{dantsin}, only for $\{1,\dots, k\}^t$ instead of 
the Boolean cube $\{0,1\}^t$.
\begin{lemma}
  For any $t,k \in \mathbb{N}$ and $0 \leq r \leq t$, there exists
  a code $\mathcal{C} \subseteq \{1,\dots,k\}^t$ of covering
  radius $r$ such that
  $$
  |\mathcal{C}| \leq\left\lceil \frac{t \ln (k) k^t}{ {t \choose r} (k-1)^r}
    \right\rceil
  $$
\label{lemma-size-of-code}
\end{lemma}
\begin{proof}
  The proof is probabilistic.  Let $m := \lceil (t \ln (k) k^t)/ ({t
    \choose r} (k-1)^r)\rceil$ and build $\mathcal{C}$ by sampling $m$
  points from $\{1,\dots,k\}$, uniformly at random and independently.
  Fix an element $w' \in \{1,\dots,k\}^t$. We calculate
  $$
  {\rm Pr}[w' \not \in \bigcup_{w \in \mathcal{C}} 
  \ball{k}{r}{w}] = 
  \left(1 - \frac{\vol{k}{t}{r}}{k^t}\right)^{|\mathcal{C}|} 
  < e^{- |\mathcal{C}|\vol{k}{t}{r} / k^t} \leq e^{-t\ln(k)} = k^{-t} \ .
  $$
  By the union bound, the probability that there is any $w' \not \in
  \bigcup_{w \in \mathcal{C}} \ball{k}{r}{w}$ is at most $k^t$ times the above
  expression, and thus smaller than $1$.  Therefore, with positive
  probability, $\mathcal{C}$ is a code of covering radius $r$.
\end{proof}

\subsection{A Deterministic Algorithm for $\PBkS{k}$}

We will now describe our deterministic algorithm. First it chooses a
sufficiently large constant $t$, depending on the $\epsilon$ in
Theorem~\ref{theorem-det-walk}, and computes a code $\mathcal{C}
\subseteq \{1,\dots,k\}^t$ of covering radius $t/k$. Since $k$ and $t$
are constants, it can afford to compute an optimal such code. We
estimate its size using Lemma~\ref{lemma-size-of-code} and the
following approximation of the binomial coefficient:
\begin{proposition}[MacWilliams, Sloane~\cite{MacWilliamsSloane77},
  Chapter 10, Corollary 9]
  For $0 \leq \rho \leq 1/2$ and $t \in \mathbb{N}$, it holds that
  $$
  {t \choose \rho t} \geq \frac{1}{\sqrt{8t\rho(1-\rho)}} 
  \left(\frac{1}{\rho}\right)^{\rho t} 
  \left(\frac{1}{1-\rho}\right)^{(1-\rho)t} 
  $$
\end{proposition}
We apply this bound with $\rho = 1/k$: 
$$
{t \choose t/k} \geq \frac{1}{\sqrt{8t}} 
k^{t/k}\left(\frac{k}{k-1}\right)^{(k-1)t/k}
 = \frac{k^t}{\sqrt{8t} (k-1)^{(k-1)t/k}} \ .
$$
Together with Lemma~\ref{lemma-size-of-code}, we obtain,
for $t$ sufficiently large:
$$
\mathcal{C} \leq \left\lceil \frac{t \ln(k) k^t}{{t \choose {t/k}}(k-1)^{t/k}}
  \right\rceil 
\leq \frac{t^2 k^t (k-1)^{(k-1)t/k}}{k^t (k-1)^{t/k}}
\leq t^2 (k-1)^{t - 2t/k} \ .
$$
The algorithm computes this constant-size code and stores it for
further use. It then calls a recursive procedure that does the real
stuff. That procedure first greedily constructs a maximal set $G$ of
pairwise disjoint unsatisfied $k$-clauses of $F$. That is, $G =
\{C_1,C_2,\dots, C_m\}$, the $C_i$ are pairwise disjoint, each $C_i$
in $G$ is unsatisfied by $\alpha$, and each unsatisfied $k$-clause $D$
in $F$ shares at least one literal with some $C_i$.\\

At this point, the algorithm considers two cases. First, if $m < t$,
it enumerates all $2^{km}$ truth assignments to the variables in $G$.
For each such assignment $\beta$, it calls $\sbs(F^{[\beta]}, \alpha,
r)$ and returns $\texttt{true}$ if at least one such call returns
$\texttt{true}$. Correctness is easy to see: At least one $\beta$
agrees with the promised assignment $\alpha^*$, and therefore
$\alpha^*$ still satisfies $F^{[\beta]}$. To analyze the running time,
observe that for any such $\beta$, the formula $F^{[\beta]}$ contains
no unsatisfied clause of size $k$. This follows from the maximality of
$G$. Therefore, Proposition~\ref{prop-short-neg} tells us that
$\sbs(F^{[\beta]},\alpha,r)$ runs in time $O^*((k-1)^r)$, and
therefore this case takes time $2^{km} O^*((k-1)^r)$. Since $m < t$, and
$t$ is a constant,
this is $O^*((k-1)^r)$.\\

The second case is more interesting: If $m \geq t$, the algorithm
chooses $t$ clauses from $G$ to form $H = \{C_1,\dots,C_t\}$, a set of
pairwise disjoint $k$-clauses, all unsatisfied by $\alpha$. At this
point, our code will come into play, but first we introduce some
notation: For $w \in \{1,\dots,k\}^t$, let $\alpha[w]$ be the assignment
obtained from $\alpha$ by flipping the value of the $w_i$\textsuperscript{th}
literal in $C_i$, for $1\leq i \leq t$. To do this, the algorithm has
to choose a fixed but arbitrary ordering on $H$ as well as on the
literals in each $C_i$. Note that $\alpha[w]$ satisfies exactly one
literal in each $C_i$, for $1 \leq i \leq t$.  Strictly speaking
$\alpha[w]$ depends not only on $w$, but also on $H$, so we should
write $\alpha[H,w]$ instead of $\alpha[w]$. However, as long as $H$ is
understood, we write $\alpha[w]$.\\

Let us give an example. Suppose $\alpha$ is the all-$0$-assignment,
$t=3$ and $H = \{(x_1 \vee y_1 \vee z_1), (x_2 \vee y_2 \vee z_2),
(x_3 \vee y_3 \vee z_3)\}$. Let $w = (2,3,3)$. Then $\alpha[w]$ is the
assignment that sets $y_1$, $z_2$, and $z_3$ to $1$ and all other
variables to $0$.\\

Consider now the promised satisfying assignment $\alpha^*$ with
$d_H(\alpha,\alpha^*) \leq r$. We define $w^* \in \{1,\dots,k\}^t$ as
follows: For each $1 \leq i \leq t$, we set $w^*_i$ to $j$ such
that $\alpha^*$ satisfies the $j$\textsuperscript{th} literal in
$C_i$. Since $\alpha^*$ satisfies at least one literal in each $C_i$,
we can do this, but since $\alpha^*$ possibly satisfies multiple
literals in $C_i$, the choice of $w^*$ is not unique. Note that in any
case
$d(\alpha[w^*], \alpha^*) = d(\alpha,\alpha^*) - t \leq r-t$.\\

We could now iterate over all $w \in \{1,\dots,k\}^t$ and call $\sbs(F,
\alpha[w], r-t)$. This would essentially be what $\sbs$ does and would
yield a running time of $O^*(k^r)$, i.e., no improvement over Dantsin
et al. Therefore, we do not do this. Instead, we let our code
$\mathcal{C}$ play its crucial role: Rather than recursing on
$\alpha[w]$ for each $w \in \{1,\dots,k\}^t$, we recurse only for each
$w \in \mathcal{C}$.  By the properties of $\mathcal{C}$, there is
some $w' \in \mathcal{C}$ such that $d_H(w',w^*) = t/k$. Observe what
happens when we go from $\alpha$ to $\alpha[w']$: For at most $t/k$
coordinates $i$, we have $w'_i \ne w^*_i$.  For those coordinates, switching
the $w'_i$\textsuperscript{th} literal of $C_i$ in the assignment
$\alpha$ {\em increases} the distance to $\alpha^*$.  On the other
hand, there are at least $t - t/k$ coordinates $i$ where $w'_i = w^*_i$,
and switching the $w'_i$\textsuperscript{th} literal of $C_i$ for such
an $i$ {\em decreases} the distance to $\alpha^*$. We conclude that
the distance increases at most $t/k$ times and decreases at least
$t-t/k$ times.  Therefore
$$
d_H(\alpha[w'], \alpha^*) \leq d_H(\alpha,\alpha^*) + t/k - (t - t/k) \leq
r - (t - 2t/k) . 
$$
Writing $\Delta := (t-2t/k)$, the procedure calls itself recursively
with $\alpha[w]$ and $r - \Delta$ for each $w \in \mathcal{C}$ and at
least one call will be successful. Let us analyze the running time: We
cause $|\mathcal{C}|$ recursive calls and decrease the complexity
parameter $r$ by $\Delta$ in each step. This is good, since
$|\mathcal{C}|$ is only slightly bigger than $(k-1)^{\Delta}$. We
conclude that the number of leaves in this recursion tree is at most
$$
|\mathcal{C}|^{r / \Delta} \leq (t^2 (k-1)^{\Delta})^{r/\Delta}
= \left( (k-1) t^{2 / \Delta} \right)^r \ .
$$
Since $t^{2 / \Delta}$ goes to $1$ as $t$ grows, the above term is,
for sufficiently large $t$, bounded by $(k-1+\epsilon)^r$. This proves
Theorem~\ref{theorem-det-walk}.
We summarize the whole procedure in Algorithm~\ref{sbf}.
\begin{algorithm}
  \caption{\sbf($k \in \mathbb{N}$, $(\le k)$-CNF formula $F$, assignment
    $\alpha$, radius $r$, code $\mathcal{C} \subseteq
    \{1,\dots,k\}^t)$}
\label{sbf}
\begin{algorithmic}[1]
  \IF{$\alpha$ satisfies $F$}
  \RETURN \texttt{true}
  \ELSIF{$r = 0$} \RETURN \texttt{false}
  \ELSE
  \STATE $G \leftarrow $ a maximal set of pairwise disjoint
  $k$-clauses of $F$ unsatisfied by $\alpha$
  \IF{$|G| < t$}
  \FOR{each assignment $\beta$ to the variables in $G$}
  \IF{$\sbs(F^{[\beta]},\alpha,r) = \texttt{true}$}
  \RETURN \texttt{true}
  \ENDIF
  \ENDFOR 
  \ELSE 
  \STATE $H \leftarrow \{C_1,\dots,C_t\} \subseteq G$
  \FOR{$w \in \mathcal{C}$}
  \IF{$\sbf(F, \alpha[H,w], r- (t-2t/k)) = \texttt{true}, \mathcal{C}$}
  \RETURN \texttt{true}
  \ENDIF
  \ENDFOR
  \ENDIF
  \ENDIF
  \RETURN \texttt{false}
\end{algorithmic}
\end{algorithm}

\section{Constraint Satisfaction Problems}

Constraint Satisfaction Problems, short CSPs, are generalizations of
SAT, allowing more than two truth values. Formally, suppose there is a
set $V = \{x_1,\dots,x_n\}$ of $n$ variables, each of which can take
on a value in $\{1,\dots,d\}$. A {\em literal} is an expression
of the form $(x_i \ne c)$ for $c \in \{1,\dots,d\}$. A constraint
is a disjunction of literals, for example
$$
(x_1 \ne 7 \vee x_2 \ne 5 \vee x_3 \ne d) \ .
$$
A CSP formula finally is a conjunction of constraints. We call it a
$(d,\le k)$-CSP formula if its variables can take $d$ values and each
constraint has at most $k$ literals. An assignment $\alpha$ to the
variables $V$ is a function $\alpha: V \rightarrow \{1,\dots,d\}$ and
can be represented as an element from $\{1,\dots,d\}^n$. We say
$\alpha$ satisfies the literal $(x_i \ne c)$ if, well, $\alpha(x_i)
\ne c$.  It satisfies a constraint if it satisfies at least one
literal in it, and it satisfies a CSP formula if it satisfies each
constraint in it. Finally, $(d,\le k)$-CSP is the problem of deciding
whether a given $(d,\le k)$-CSP formula has a satisfying assignment.  Note
that $(2,k)$-CSP is the same as $k$-SAT. Also, $(d, \le k)$-CSP is
NP-complete except the following three cases: (i) $d=1$, (ii) $k=1$,
(iii) $d=k=2$. Cases (i) and (ii) are trivial problems, and (iii) is
$2$-SAT, which is solvable in polynomial time (well-known, not
difficult to show, but still not trivial).\\

For the cases where $(d,\le k)$-CSP is NP-complete, what can we do?
Iterating through all $d^n$ assignments constitutes an algorithm
solving $(d,k)$-CSP in time $O^*(d^n)$. Sch\"oning's
algorithm~\cite{Schoening99} is much faster:

\begin{theorem}[Sch\"oning~\cite{Schoening99}]
  There is a randomized Monte-Carlo algorithm solving
  $(d,\le k)$-CSP in time 
  $$
  O^*\left( \left(\frac{d(k-1)}{k}\right)^n \right) \ .
  $$
\label{theorem-sch-csp}
\end{theorem}

Again, for $d=2$ this is the running time of $\schoening$ for $k$-SAT. 
In his original paper~\cite{Schoening99}, Sch{\"o}ning describes how his
algorithm seamlessly generalizes to arbitrary domain sizes $d \ge 2$: in
each correction step, after a variable to reassign has been selected at random,
another random choice is made among the $d-1$ values it may be changed to.
The subsequent analysis in~\cite{Schoening99} also extends to this case.

However, there is a more direct way to reduce the $(d,\le k)$-CSP for
$d>2$ to the Boolean problem which is then able to use any $k$-SAT
algorithm as a black box: we simply select for each variable,
uniformly at random and independently from the other variables, $2$
out of the $d$ possible values in the domain. Any satisfying
assignment survives this restriction with probability exactly
$(2/d)^n$ and thus any $k$-SAT algorithm with success probability
$p^n$ generalizes to a $(d,\le k)$-CSP algorithm with success
probability $(2p/d)^n$. When plugging in $\schoening$ for $k$-SAT, we
obtain Theorem~\ref{theorem-sch-csp}.

In order to generalize our deterministic variant to arbitrary domain
sizes, we will choose the simple route and derandomize the
aforementioned reduction instead of trying to rework the whole
analysis from the previous section, with the additional advantage that
the result scales for any further improvement on the running time for
deterministic $k$-SAT.

\begin{theorem}
  There exists a deterministic algorithm having running time
  $O^*((d/2)^n)$ which takes any $(d,\le k)$-CSP $F$ over $n$
  variables and produces $l = O^*((d/2)^n)$ Boolean $(\le k)$-CNF
  formulas $\{G_i\}_{1 \le i \le l}$ such that $F$ is satisfiable if
  and only if there exists some $i$ such that $G_i$ is satisfiable.
\label{theorem-det-reduction}
\end{theorem}

Using the $k$-SAT algorithm we developed in the previous section, we
then immediately get the derandomization of
Theorem~\ref{theorem-sch-csp}.

\begin{corollary}
  For every $\epsilon > 0$, there is a deterministic algorithm solving
  $(d,\le k)$-CSP in time
  $$
  O^*\left( \left(\frac{d(k-1)}{k}+\epsilon\right)^n \right) \ .
  $$
\label{corollary-det-csp}
\end{corollary}

\begin{proof}[Proof of Theorem~\ref{theorem-det-reduction}]
  We start with a useful definition. A {\em $2$-box} in
  $\{1,\dots,n\}^d$ is a set of the form $B := P_1 \times \dots \times
  P_n$, where $P_i \subseteq \{1,\dots,d\}$ and $|P_i| = 2$.  A
  $2$-box can be seen as a subcube of $\{1,\dots,d\}^n$ of side length
  $2$ and full dimension. A {\em random} $2$-box is a $2$-box sampled
  uniformly at random from all $2$-boxes in $\{1,\dots,d\}^n$: This
  can be done by sampling each $P_i$ independently, uniformly at
  random from all ${d \choose 2}$ pairs in $\{1,\dots,d\}$. As
  mentioned above, the probability that any fixed satisfying
  assignment of $F$ lies within
  a random $2$-box is $(2/d)^n$.\\
  
  In order to derandomize this technique, we need to deterministically 
  cover $\{1,\dots,d\}^n$  with $2$-boxes, in a fashion very similar to the covering
  codes used by Dantsin et al.~\cite{dantsin}:
  \begin{lemma}
    Let $d,n \in \mathbb{N}$. There is a set $\mathcal{B}$
    of $2$-boxes in $\{1,\dots,d\}^n$ such that
    $$
    \bigcup_{D \in \mathcal{B}} B = \{1,\dots,d\}^n 
    $$
    and 
    $$
    |\mathcal{B}| \leq \left(\frac{d}{2}\right)^n \poly(n) \ .
    $$
    Furthermore, $\mathcal{B}$ can be constructed in time
    $O(|\mathcal{B}|)$.    
    \label{lemma-2-boxes}
  \end{lemma}
  Given this lemma, our algorithm is complete: It first constructs
  such a suitably small set $\mathcal{B}$ of $2$-boxes, and then, for
  each $2$-box $P_1 \times \dots \times P_n = B \in \mathcal{B}$,
  outputs a $(\le k)$-CNF formula arising from $F$ by restricting the
  domain of the $i$\textsuperscript{th} variable to the values in
  $P_i$.  This finishes the proof of the theorem.
\end{proof}

It remains to prove the lemma.
\begin{proof}[Proof of Lemma~\ref{lemma-2-boxes}]
  Note that if $d$ is an even number, the proof is easy. For $1 \leq j
  \leq d/2$, define $P^{(j)} = \{2j-1, 2j\}$. Each element $w \in
  \{1,\dots,d/2\}^n$ defines the $2$-box
  $$
  B_w := P^{(w_1)} \times \dots \times P^{(w_n)} 
  $$
  and clearly
  $$
  \bigcup_{w \in \{1,\dots,d/2\}^n} B_w = \{1,\dots,d\}^n \ .
  $$
  The difficulty arises if $d$ is odd. As Dantsin et
  al.~\cite{dantsin}, we first show the existence of a suitable set of
  $2$-boxes, and then use a block construction and an approximation
  algorithm to obtain a construction.

  \begin{lemma}
    For any $n, d\in \mathbb{N}$, there is a set $\mathcal{B}$ of
    $2$-boxes such that $|\mathcal{B}| \leq \left\lceil n \ln(d)
      (d/2)^n\right\rceil$ such that $\bigcup_{B \in \mathcal{B}} =
    \{1,\dots,d\}^n$.
    \label{lemma-existance-box-cover}
  \end{lemma}
  \begin{proof}
    The proof works exactly like the proof of
    Lemma~\ref{lemma-size-of-code}.  We sample $\left\lceil n \ln(d)
      (d/2)^n\right\rceil$ many $2$-boxes independently, uniformly at
    random and show that with positive probability, the resulting set
    has the desired properties.
  \end{proof}
  To prove Lemma~\ref{lemma-2-boxes}, we have to derandomize the
  probabilistic argument we have just seen. For this, we choose a
  sufficiently large constant $b$, set $n' := n/b$ and construct an
  instance of \textsc{Set-Cover}: The ground set is
  $\{1,\dots,d\}^{n'}$ and the sets are all $2$-boxes therein, of
  which there are ${d \choose 2}^{n'} \leq d^{2n'}$.  We know from
  Lemma~\ref{lemma-existance-box-cover} that there is a cover of
  $2$-boxes of size $\left\lceil n \ln(d) (d/2)^n\right\rceil$.  There
  is a greedy algorithm for \textsc{Set-Cover} (see
  Hochbaum~\cite{Hochbaum} for example) achieving an approximation
  ratio of $O(\log N)$, where $N$ is the size of the ground set. Since
  in our case $\log N = \log(d^{n'}) = O(n')$, this algorithm will
  give us a set $\mathcal{B}$ of $2$-boxes covering
  $\{1,\dots,d\}^{n'}$ of size
  $$
  |\mathcal{B}| \in O\left( (n')^2 \left(\frac{d}{2}\right)^{n'} \right) \ .
  $$
  How much time do we need to construct $\mathcal{B}$? The greedy
  algorithm is polynomial in the size of its instance, which is
  $O(d^{2n'})$, thus it takes time $O(d^{2Cn'})$ for some constant
  $C$. By choosing $b$ large enough, we can make sure that $d^{2Cn'} =
  d^{2Cn/b}$ is smaller than the running time we are aiming at.
  Finally, we obtain a set of $2$-boxes in $\{1,\dots,d\}^n$ by
  ``concatenating'' the boxes in $\mathcal{B}$: We identify a tuple
  $(B_1,\dots,B_b) \in \mathcal{B}^b$ with the $2$-box $B_1 \times
  \dots \times B_b$, and therefore $\mathcal{B}^b$ is a set of
  $2$-boxes covering $\{1,\dots,n\}$, and
  $$
  |\mathcal{B}|^b \leq 
  O\left( \left((n')^2 \left(\frac{d}{2}\right)^{n'}\right)^b \right)
  = O\left(   \left(\frac{n}{b}\right)^{2b} \left(\frac{d}{2}\right)^n \right)  
  = \left(\frac{d}{2}\right)^n \poly(n)
 \ .
  $$
  This is a set of $2$-boxes covering $\{1,\dots,d\}^n$ of the
  desired size, finishing the proof of Lemma~\ref{lemma-2-boxes}.
\end{proof}

\section*{Acknowledgments}

We thank our supervisor Emo Welzl for continuous support. The second
author thanks Konstantin Kutzkov for the fruitful collaboration
on~\cite{kutzkov-scheder}.

\bibliographystyle{abbrv}
\bibliography{refs}

\end{document}